\documentclass{amsart}                     
\usepackage{graphicx}
\usepackage{bbm}
\usepackage{amsmath}
\usepackage{amssymb}
\usepackage[english]{babel}
\usepackage{relsize}
\usepackage{amsthm}
\usepackage{mathrsfs}
\usepackage[latin1]{inputenc}
\usepackage{subfigure}
\usepackage{fullpage}
\usepackage{color}
\usepackage{comment}
\usepackage{tikz}
\usepackage{marginnote}%
\DeclareMathOperator{\diag}{diag}

\newcommand{\TT}{\mathbb{T}}

\newcommand{\RR}{\mathbb{R}}
\newcommand{\ZZ}{\mathbb{Z}}

\newcommand{\BB}{\mathcal{B}}

\newtheorem{theorem}{Theorem}
\newtheorem{problem}[theorem]{Problem}
\newtheorem{remark}[theorem]{Remark}
\newtheorem{lemma}[theorem]{Lemma}
\newtheorem{proposition}[theorem]{Proposition}
\newtheorem{definition}[theorem]{Definition}

\begin{document}

\title[Orbit determination for hyperbolic maps]{Asymptotic behaviour of orbit determination for hyperbolic maps}
\thanks{This work was partially supported by the National Group of Mathematical Physics (GNFM-INdAM) through the project ``Orbit Determination: from order to chaos'' (Progetto Giovani 2019). This research is also part the authors' activity within the UMI-DinAmicI community (www.dinamici.org) and the GNFM-INdAM, Italy.}


\author{Stefano Mar\`o}
\address{Dipartimento di Matematica, Universit\`a di Pisa, 
  Largo B. Pontecorvo 5, 56127 Pisa, Italy}
\email{stefano.maro@unipi.it}

\author{Claudio Bonanno}
\address{Dipartimento di Matematica, Universit\`a di Pisa,
Largo B. Pontecorvo 5, 56127 Pisa, Italy}
\email{claudio.bonanno@unipi.it}      

\maketitle

\begin{abstract}
   We deal with the orbit determination problem for
   hyperbolic maps.  The problem consists in determining the initial
   conditions of an orbit and, eventually, other parameters of the model from some observations.
   We study the behaviour of the confidence region in the case of 
   simultaneous increase of the number of observations and the time
   span over which they are performed. More precisely, we describe the
   geometry of the confidence region for the solution, distinguishing
   whether a parameter is added to the estimate of the initial
   conditions or not. We prove that the inclusion of a dynamical
   parameter causes a change in the rate of decay of the
   uncertainties, as suggested by some known numerical evidences.

\end{abstract}

\section{Introduction}
\label{intro}
This paper is concerned with the behaviour of the confidence region
coming from an orbit determination process as the number of
observation increases.

We recall that orbit determination consists of recovering information
on some parameters (initial conditions or dynamical parameters) of a
model given some observations and goes back to Gauss
\cite{gauss}. The solution, called nominal solution, relies on the
least squares algorithm and the confidence region summarises the
uncertainties coming from the intrinsic errors in the observational
process.

The problem under investigation is suggested by the numerical results
in \cite{ssm,smil} where some estimates are given for the cases of a
map depending on a parameter and presenting both ordered and chaotic
zones: the Chirikov standard map \cite{chir} (see also \cite{sm,cell} for its importance in Celestial Mechanics).  The authors of \cite{ssm,smil} constructed
the observations by adding some noise to a true orbit of the
map. Then, they set up an orbit determination process to recover the
true orbit and observed the decay of the uncertainties as the number of observations grows.  The experiments show that the result crucially depends
on the dynamics and on whether the parameter is included in the orbit
determination process or not. More precisely, if the observations come
from an ordered zone (an invariant curve), then the uncertainties
decrease polynomially both if the parameter is included or not. This
behaviour was analytically proven to be true at least in the case
where only the initial conditions are estimated in \cite{maro1}, using
KAM techniques.

The numerical results coming from the chaotic case are more
delicate. From the practical point of view, the problem of the
so-called computability horizon occurs. This prevents orbit
determination from being performed if the time span of the
observations is too large and more sophisticated techniques must be
employed, such as the multi-arc approach \cite{ssm}.  Moreover, at
least until the computability horizon, the uncertainties on the sole
initial conditions decrease exponentially, while a polynomial decay is
observed when the parameter is included in the orbit determination
process.

In this paper we give an analytical proof of this last result. We will consider a
class of hyperbolic maps depending on a
parameter. The existence of chaotic orbits, in the future or in the past, for these systems is given by the fact that all the
Lyapunov exponents are supposed to be non zero. Despite the above
described practical problem in computing a solution, we will always
suppose that the least squares algorithm converges and gives a nominal
solution. Hence, the estimates on the decay of the uncertainty are
given asymptotically as the number of observations goes to infinity. To
state the result we recall that the confidence region is an ellipsoid
in the space of the fit parameters and the uncertainties strictly
depend on the size of the axes of such an ellipsoid.

We will prove that, in the case of estimating only the initial
conditions, there exists a full measure set of possible nominal
solutions for which all the axes of the related confidence ellipsoid
decay exponentially.
On the other hand, if the parameter is included in the orbit
determination process then there exists a full measure set of initial
conditions and parameters for which the related confidence ellipsoid
has an axis that decays strictly slower than exponentially.

Our analytical results are consistent with the numerical
results in \cite{ssm,smil}. In the case of estimating the parameter we
cannot prove polynomial decay of the uncertainties,
however we will show that this occurs for a class of affine maps depending on a parameter. Perturbed automorphisms of the torus are representative of this class, including the famous Arnold's Cat Map.

We conclude stressing the fact that chaotic orbit determination is a
challenge for both space mission and impact monitoring. Actually, the
accurate determination of orbits of chaotic NEOs is essential in the
impact monitoring activity \cite{mv}. Another interesting case is
given by satellites. When their operating life finishes, they are left
without control in safe orbits, governed only by the natural
forces. It has been noticed that many parts of this region are chaotic
\cite{ros}, due to the perturbed motion of the Moon. It is then
important to track and determine orbits of non-operating satellites
that could crash into operating ones. Finally, the targets
of many space missions include the determination of some unknown
parameter. Typical examples are the ESA/JAXA BepiColombo mission to
Mercury, the NASA JUNO and ESA JUICE missions to Jupiter that are
performed in a chaotic environment \cite{lm}.

Moreover, this results are related to a conjecture posed by Wisdom in
1987 (see \cite{wis}). Discussing the chaotic rotation state of
Hyperion, it was proposed that ``the knowledge gained from
measurements on a chaotic dynamical system grows exponentially with
the time span covered by the observations''. In particular, this was related to
the information on dynamical parameters like the moments of inertia
ratios.

The paper is organised as follows. In Section \ref{statement} we adapt
the general description of the problem given in \cite{maro1} to our
situation and state our main results. Moreover we briefly discuss our results as compared with the numerical simulations in \cite{ssm,smil}. Section \ref{sec:proofA} is
dedicated to the proof of the result concerning the estimation of the
sole initial conditions, while in Section \ref{sec:proofB} is
dedicated to the proof of the results in the case the parameter is included. Section \ref{sec:example} is dedicated to the study of a concrete example, and our conclusions are given in Section \ref{sec:conclusions}.

\section{Statement of the problem and main results}\label{statement}

\subsection{Notation and preliminaries on Lyapunov exponents}\label{notation}
The main tool of our approach to the orbit determination problem is the notion of Lyapunov exponents of a differentiable map, of which we now recall the definition and the main properties as stated in the Oseledets Theorem. Let $f:X\to X$ be a diffeomorphism of an $d$-dimensional differentiable manifold $X$ endowed with a $\sigma$-algebra $\BB$ and an $f$-invariant probability measure $\mu$. We recall that a measure is $f$-invariant if $\mu(E) =\mu(f^{-1}(E))$ for all $E\in \BB$. One can work in the local charts of $X$, so we use the notation $f(x) = ((f)_1(x),\dots,(f)_d(x))$ for the components of $f$, and denote the Jacobian matrix of $f$ by  
\begin{equation}\label{Fk}
  F(x) =\left(
  \begin{array}{ccc}
    \frac{\partial (f)_1}{\partial x_1}(x) & \dots & \frac{\partial (f)_1}{\partial x_d}(x) \\
    \vdots &  & \vdots \\
    \frac{\partial (f)_d}{\partial x_1}(x) & \dots & \frac{\partial (f)_d}{\partial x_d}(x)
  \end{array}
  \right)
\end{equation}
and, for $n\in\ZZ$, we denote by $F^n(x)$ the Jacobian matrix of $f^n$. By the chain rule it can be written as
\begin{equation}\label{Fnk}
  F^n(x) =\left\{
  \begin{split}
   &  F(f^{n-1}(x))F(f^{n-2}(x))\cdots F(x) \quad &\mbox{for }n\geq 1, \\
    & \mathbbm{1} \quad &\mbox{for }n = 0, \\ 
    & F^{-1}(f^{n}(x))F^{-1}(f^{n+1}(x))\cdots F^{-1}(f^{-1}(x)) \quad &\mbox{for }n<0. 
    \end{split}
  \right.
  \end{equation}
Let us now introduce the Lyapunov exponents of $f$. Suppose that
\begin{equation}\label{L1}
\int_X \log \|F^{\pm 1} \| d\mu <\infty.
 \end{equation}
Given $x\in X$ and a vector $v\in T_xM$, let us define 
\begin{equation}\label{limiti-lyap}
\begin{aligned}
&  \overline{\gamma}_\pm (x,v) := \limsup_{n\to\pm\infty}\frac{1}{|n|}\log \left|F^n(x)v\right|, \\
&  \underline{\gamma}_\pm (x,v) := \liminf_{n\to\pm\infty}\frac{1}{|n|}\log \left|F^n(x)v\right|.
\end{aligned}
\end{equation}
If $\overline{\gamma}_\pm (x,v) =\underline{\gamma}_\pm (x,v)$ we use the notation 
\[
\gamma_\pm(x,v) := \overline{\gamma}_\pm (x,v) =\underline{\gamma}_\pm (x,v).
\]
In principle, the limits $\gamma_\pm(x,v)$ depends on $x$ and $v$. The following version of the classical theorem by Oseledets \cite{ose} (see also \cite{rag,rue}) gives an answer to this problem. The interested reader can find more details on Lyapunov exponents in \cite{bar_pes}.

\begin{theorem}[Oseledets]\label{osel}
Let $f:X\to X$ be a measure-preserving diffeomorphism of an $d$-dimensional differentiable manifold $X$  endowed with a $\sigma$-algebra $\BB$ and an $f$-invariant probability measure $\mu$, and assume that the Jacobian matrix $F(x)$ satisfies \eqref{L1}. Then for $\mu$-almost every $x\in X$ there exist numbers
  \[
\gamma_1(x)<\gamma_2(x)<\dots<\gamma_{r(x)}(x)
  \]
  with $r(x)\leq d$ and a decomposition
  \[
T_xX = E_1(x)\oplus E_2(x)\oplus \dots\oplus E_{r(x)}(x)
\]
such that
\begin{itemize}
\item[(i)] for every $v\in E_i(x)$, $\limsup$ and $\liminf$ in \eqref{limiti-lyap} coincide and $\gamma_+(x,v) = -\gamma_-(x,v) = \gamma_i(x)$;
\item[(ii)] the functions $x\mapsto r(x),\gamma_i(x),E_i(x)$ are measurable and $f$-invariant;
\item[(iii)] if $f$ is ergodic, then $r(x),\gamma_i(x),E_i(x)$ are $\mu$-a.e. constant;
\item[(iv)] for $\mu$-a.e. $x\in X$, the matrix
  \[
  \Lambda(x) := \lim_{n\to\infty}\left[(F^n(x))^TF^n(x)\right]^{1/2n}
  \]
exists and $\exp \gamma_1(x),\dots,\exp \gamma_{r(x)}(x)$ are its eigenvalues. 
  \end{itemize}
\end{theorem}

\begin{definition} \label{def:lyap-exp}
The numbers $\gamma_1(x),\dots,\gamma_{r(x)}(x)$ given in Theorem \ref{osel} are the \emph{Lyapunov exponents of $f$ at $x$}, and for each $\gamma_i(x)$, $i=1,\dots,r(x)$, the dimension of the corresponding vector space $E_i(x)$ is called the \emph{multiplicity} of the exponent.
\end{definition}

Without loss of generality one can assume that the diffeomorphism $f$ is ergodic, so that the Lyapunov exponents and their multiplicities do not depend on $x$. The case of non-ergodic maps can be treated by the standard procedure of ergodic decomposition, obtaining similar results depending on the ergodic component to which the initial condition $x$ belongs.

\begin{definition}\label{def:hyp-map}
A diffeomorphism $f:X\to X$ is called \emph{hyperbolic} if it has no vanishing Lyapunov exponents.
\end{definition}

In the particular case of Hamiltonian maps, or of maps preserving the volume form of a manifold, one can easily deduce that hyperbolic maps necessarily have positive Lyapunov exponents, so are \emph{chaotic}. Moreover, by Theorem \ref{osel}-(i), it follows that for a hyperbolic map either $f$ or $f^{-1}$ is chaotic. This is an important remark for our main results.

In the following we consider diffeomorphisms depending on a parameter $k\in K\subset \RR$ and use the notation $f_k:X\to X$. We also assume that the dependence on $k$ is differentiable and that the probability measure $\mu$ of the manifold $X$ is $f_k$-invariant and the map is ergodic for all $k\in K$. The Jacobian matrix of $f_k^n$ with respect to $x$ for $n\in\ZZ$ is denoted by $F^n_k(x)$ and can be written as in \eqref{Fk} and \eqref{Fnk}. Assuming that \eqref{L1} is satisfied, we can apply Oseledets Theorem to $f_k$ for all $k\in K$ and find its Lyapunov exponents at $\mu$-a.e. $x\in X$.

Since the map $f_k$ is differentiable also with respect to the parameter $k$, we can also consider the Jacobian matrix of $f_k$ with respect to $(x,k)$ which is denoted by
\begin{equation}\label{Fktilde}
\tilde{F}_k(x) =\left(
  \begin{array}{ccc|c}
    \frac{\partial (f_k)_1}{\partial x_1}(x) & \dots & \frac{\partial (f_k)_1}{\partial x_d}(x) & \frac{\partial (f_k)_1}{\partial k}(x) \\
    \vdots &  & \vdots & \vdots \\
    \frac{\partial (f_k)_d}{\partial x_1}(x) & \dots & \frac{\partial (f_k)_d}{\partial x_d}(x) & \frac{\partial (f_k)_d}{\partial k}(x)
  \end{array}
  \right)
\end{equation}
Analogously, for $n\in\ZZ$ the Jacobian matrix of $f_k^n$ with respect to $(x,k)$ is denoted by $\tilde{F}^n_k(x)$ and can be written as
\begin{equation}\label{Fnktilde}
  \tilde{F}^n_k(x) =
  \left(
  \begin{array}{ccc|c}
     &  &  & \frac{\partial (f^n_k)_1}{\partial k}(x) \\
    \multicolumn{3}{c|}{F^n_k(x)}  & \vdots \\
     &  & & \frac{\partial (f^n_k)_d}{\partial k}(x)
  \end{array}
  \right)
\end{equation}

\subsection{Statement of the problem}
A general statement of the problem can be found in \cite{maro1}. For the sake of completeness, here we recall and adapt it to the present notations.

Consider a map $f_k$ as in the previous section. Given an initial condition $x$, its orbit is completely determined by the iterations $f_k^n(x)$ for $n\in \ZZ$. Instead let's suppose that we have been observing the evolution of the state of a system modelled by $f_k$ and that we have got the observations $(\bar{X}_n)$ for $|n|\leq N$.
Following \cite{mg2010} we set up an orbit determination process to determine the unknown parameters. We consider two different scenarios.
\begin{itemize}
\item[(A)] Only the initial conditions $x$ are unknown.
\item[(B)] Both the initial conditions $x$ and the parameter $k$ are unknown.
\end{itemize}
In both cases, we search for the values of the parameters that best approximate, in the least squares sense, the given observations. We first define the \emph{residuals} as
\begin{equation}\label{def_res}
\begin{aligned}
&\xi_{n,k}(x):=\bar{X}_n-f^n_k(x) &\text{for case (A)}, \\
&\tilde{\xi}_n(x,k):=\bar{X}_n-f^n_k(x) &\text{for case (B)}.
\end{aligned}
\end{equation}
We stress that, even if the expressions coincide, in case (A) the residuals are defined in terms of a fixed $k$, whereas in case (B) the value of $k$ is to be determined. 

\noindent Subsequently, we call the \emph{least squares solution} $x_0$ in case (A), or $(x_0,k_0)$ in case (B), the (local) minimiser of the \emph{target function}
\begin{equation}\label{def_target}
\begin{aligned}
& Q_k(x) := \frac{1}{2N+1}\sum_{|n|\leq N}\xi_{n,k}(x)^T\xi_{n,k}(x) &\text{for case (A)}, \\
& \tilde{Q}(x,k) := \frac{1}{2N+1}\sum_{|n|\leq N}\tilde{\xi}_n(x,k)^T\tilde{\xi}_n(x,k) &\text{for case (B)}.
\end{aligned}
\end{equation}
We will not be concerned with the existence and computation of the
minima. This is a very delicate task, solved via iterative schemes
such as the Gauss-Newton algorithm and the differential
corrections. These algorithms crucially depend on the choice of the
initial conditions. See \cite{gbm}, \cite{mbbg} for some recent
results on this topic for the asteroid and space debris cases.
For the case of chaotic maps that we study in this paper, the problem is considered in \cite{ssm,smil} where computational problems that occur for large $N$ are treated with advanced techniques. In the following of this paper, we assume that the least squares solution $x_0$, or $(x_0,k_0)$, exists and we refer to it as the \emph{nominal solutions}.

In general the observations $(\bar{X}_n)$ contain errors, hence values of $x$, or of $(x,k)$, that make the target function slightly bigger than the minimum $Q_k(x_0)$, or $\tilde{Q}(x_0,k_0)$, are acceptable. This leads to the definition of the \emph{confidence region} as
\[
\begin{aligned}
&\mathcal{Z}_k =\left\{ x\in X \: : \: Q_k(x) \leq Q_k(x_0) +\frac{\sigma^2}{2N+1}   \right\} &\text{for case (A)}.\\
&\tilde{\mathcal{Z}} =\left\{ (x,k)\in X\times \RR \: : \: \tilde{Q}(x,k) \leq \tilde{Q}(x_0,k_0) +\frac{\sigma^2}{2N+1} \right\} &\text{for case (B)},
\end{aligned}
\]
where $\sigma>0$ is an empirical parameter chosen depending on statistical properties and bounds the acceptable errors; the value of $\sigma$ is irrelevant for our purposes, hence in the next sections we will set $\sigma=1$. Expanding the target functions $Q_k(x)$ and $\tilde{Q}(x,k)$ at the corresponding nominal solution up to second order we get, using the notation introduced in \eqref{Fnk} and \eqref{Fnktilde} and working in local charts on $X$ so that we use the notation $x-x_0$ for a vector in $\RR^d$,
\begin{multline*}
Q_k(x) \sim Q_k(x_0)\,  + \\
\frac{1}{2N+1}
\left(\begin{array}{l}
  x-x_0
\end{array}
\right)^T
\sum_{|n|\leq N}\left[ F_k^n(x_0)^TF_k^n(x_0) + d(F_k^n(x_0))\xi_{n,k}^T(x_0)\right] \left(\begin{array}{l}
  x-x_0
\end{array}
\right)
\end{multline*}
and
\begin{multline*}
\tilde{Q}(x,k) \sim \tilde{Q}(x_0,k_0)\,  + \\
\frac{1}{2N+1}
\left(\begin{array}{l}
  x-x_0\\
  k-k_0
\end{array}
\right)^T
\sum_{|n|\leq N}\left[\tilde{F}_{k_0}^n(x_0)^T\tilde{F}_{k_0}^n(x_0) + d(\tilde{F}_{k_0}^n(x_0))\tilde{\xi}_{n}^T(x_0,k_0)\right] \left(\begin{array}{l}
  x-x_0\\
  k-k_0  
\end{array}
\right).
\end{multline*}

Under the hypothesis that the residuals corresponding to the nominal solution are small, we can neglect the terms $\xi_{n,k}^T(x_0)$ and $\tilde{\xi}_{n}^T(x_0,k_0)$. Then, we define the \emph{normal matrices} as
\begin{equation}\label{normmat}
\begin{aligned}
& C_{N,k}(x) := \sum_{|n|\leq N}F_k^n(x)^TF_k^n(x) &\text{for case (A)},\\
& \tilde{C}_{N}(x,k) := \sum_{|n|\leq N}\tilde{F}_{k}^n(x)^T\tilde{F}_{k}^n(x) &\text{for case (B)},
\end{aligned}
\end{equation}
and the associated \emph{covariance matrices} as
\[
\begin{aligned}
&  \Gamma_{N,k}(x):= \left[C_{N,k}(x)\right]^{-1} &\text{for case (A)}, \\
&  \tilde{\Gamma}_{N}(x,k) := \left[\tilde{C}_{N}(x,k)\right]^{-1} &\text{for case (B)}.
\end{aligned}
\]
Note that the matrices $C_{N,k}(x)$ and $\Gamma_{N,k}(x)$ defined for case (A) are $d\times d$, while the matrices $\tilde{C}_{N}(x,k)$ and $\tilde{\Gamma}_{N}(x,k)$ defined for case (B) are $(d+1)\times (d+1)$.  Moreover, the normal matrices are symmetric and positive definite since $f_k$ is a diffeomorphism and the operators $F_k(x)$ and $\tilde{F}_k(x)$ have maximum rank.  
Hence, the confidence regions can be approximated by the \emph{confidence ellipsoids} given by
\begin{equation} \label{conf-ellip-A}
\mathcal{E}_{N,k}(x_0)
:=\left\{ x\in\RR^d \: : \:
\left(\begin{array}{l}
  x-x_0
\end{array}
\right)^T
C_{N,k}(x_0)
\left(
\begin{array}{l}
  x-x_0
\end{array}
\right)
\leq \sigma^2  \right\}
\end{equation}
for case (A), and by 
\begin{equation} \label{conf-ellip-B}
\tilde{\mathcal{E}}_{N}(x_0,k_0):=\left\{ (x,k)\in\RR^{d+1} \: : \:
\left(\begin{array}{l}
  x-x_0\\
  k-k_0
\end{array}
\right)^T
\tilde{C}_{N}(x_0,k_0)
\left(
\begin{array}{l}
  x-x_0 \\
  k-k_0  
\end{array}
\right)
\leq \sigma^2  \right\}
\end{equation}
for case (B). The covariance matrices $\Gamma_{N,k}(x)$ and $\tilde{\Gamma}_{N}(x,k)$ describe the corresponding confidence ellipsoids $\mathcal{E}_{N,k}$ and $\tilde{\mathcal{E}}_{N}$ since the axes of the ellipsoids are proportional to the square root of the eigenvalues of the corresponding matrix and are directed along the corresponding eigenvectors. Since the matrix $\Gamma_{N,k}(x)$ is positive definite, its eigenvalues are all real and positive and we denote them by
\[
0<\lambda_{N,k}^{(1)}(x)\leq\dots\leq\lambda_{N,k}^{(d)}(x).
\]
Analogously, we denote by
\[
0<\tilde{\lambda}_{N}^{(1)}(x,k)\leq\dots\leq\tilde{\lambda}_{N}^{(d+1)}(x,k)
\]
the eigenvalues of $\tilde{\Gamma}_{N}(x,k)$.

The regions $\mathcal{E}$ represent the uncertainty of the nominal solution: the values inside $\mathcal{E}$ are acceptable and the projections of $\mathcal{E}$ on the axes, represent the (marginal) uncertainties of the coordinates. See Figure \ref{figure:ellipse}.

We remark that the normal and covariance matrices also have a probabilistic interpretation, see \cite{mg2010}. 

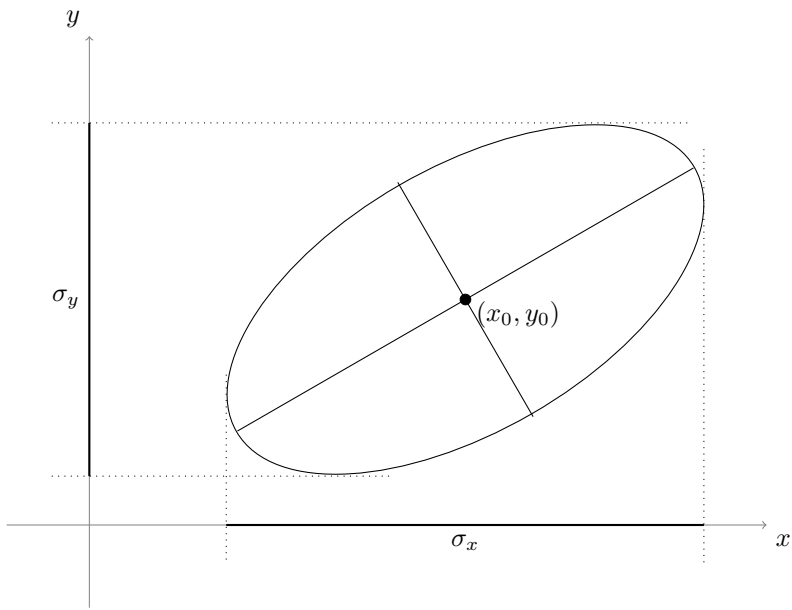
\begin{figure}[h!]
   \centering
  \begin{tikzpicture}
    \draw [help lines, ->]  (-1.1,0) -- (9,0);
    \draw [help lines, ->] (0,-1.1) -- (0,6.5);
    \draw [rotate around={30:(5,3)}] (5,3) ellipse (100pt and 50pt);
    \draw [rotate around={30:(5,3)}] (1.5,3) -- (8.5,3);
    \draw [rotate around={30:(5,3)}] (5,1.2) -- (5,4.8);
    \draw [dotted] (1.82,2) -- (1.82,-0.5);
    \draw [dotted] (8.17,5) -- (8.17,-0.5);
    \draw [dotted] (-0.5,0.65) -- (4,0.65);
    \draw [dotted] (-0.5,5.35) -- (8,5.35);
    \draw [thick] (1.82,0) -- (8.17,0);
    \draw [thick] (0,0.65) -- (0,5.35);
    \node [below right] at (9,0) {$x$};
    \node [left] at (0,3) {$\sigma_y$};
    \node [below] at (5,0) {$\sigma_x$};
    \node  at (5.7,2.8) {$(x_0,y_0)$};
    \filldraw (5,3) circle[radius=2pt];
    \node [above left] at  (0,6.5) {$y$};
  \end{tikzpicture}
  \caption{The confidence ellipse for the nominal value $(x_0,y_0)$. The values $\sigma_x,\sigma_y$ represent the marginal uncertainties of $x_0,y_0$ respectively, and depend on the value of $\sigma$.}
   \label{figure:ellipse}
\end{figure}

From the point of view of the applications, (e.g. impact monitoring \cite{mv}), it is of fundamental importance to know the shape and the size of the confidence ellipsoid $\mathcal{E}$. Hence, the question that we here address, stated in a broad sense, is the following:

\begin{problem} \label{p1}
Given a map $f_k$ as in Section \ref{notation} and a nominal solution of the associated orbit determination process, describe the confidence ellipsoids for large $N$ in cases (A) and (B).
\end{problem}

\begin{remark}
The solution of the problem passes through the computation of the eigenvalues of the covariance matrices for large $N$. Note that they crucially depend on the dynamics, since we have to compute the linearisation of the system along an orbit.
\end{remark}

\subsection{Main results} \label{sec-results}
In this paper we consider Problem \ref{p1} for hyperbolic maps. We now state and comment the results, giving the proofs in Sections \ref{sec:proofA} and \ref{sec:proofB}.

For all $k\in K\subset \RR$, let $f_k :X\to X$ be an ergodic hyperbolic diffeomorphism of a $d$-dimensional manifold $X$ with $f_k$-invariant probability measure $\mu$, and assume that $f_k$ satisfies \eqref{L1}.

For case (A) we have the following result
\begin{theorem}\label{mainA} 
Let $\gamma_1,\dots,\gamma_r$ be the Lyapunov exponents of $f_k$, and let 
\[
0< \gamma_*:= \min_i\, |\gamma_i| \le \gamma^*:= \max_i |\gamma_i|\, .
\] 
For $\mu$-almost every $x\in X$ the eigenvalues $\lambda_{N,k}^{(i)}(x)$ of  $\Gamma_{N,k}(x)$ satisfy 
\begin{equation}\label{thesisA}
-2\gamma^* \le \liminf_{N\to +\infty}\frac{\log \lambda_{N,k}^{(i)}(x)}{N} \le \limsup_{N\to +\infty}\frac{\log \lambda_{N,k}^{(i)}(x)}{N} \le -2 \gamma_*
\end{equation}
for every  $i=1,\dots,d$.  
\end{theorem}

Theorem \ref{mainA} shows that the axes of the confidence ellipsoid defined in \eqref{conf-ellip-A} shrink exponentially fast with the number of observation. In fact the lengths of the axes of $\mathcal{E}_{N,k}(x)$ are the square roots of the eigenvalues of the corresponding covariance matrix $\Gamma_{N,k}(x)$. Hence the exponential rate of decay of the uncertainties is controlled by the Lyapunov exponents of the orbit corresponding to the nominal solution. 

We now show how the result changes in case (B). We prove the following

\begin{theorem}\label{mainB}
For $\mu$-almost every $x\in X$ the largest eigenvalue $\tilde{\lambda}_{N}^{(d+1)}(x,k)$ of $\tilde{\Gamma}_{N}(x,k)$ is a positive number which decreases with $N$ and satisfies
\[
\lim_{N\to +\infty} \frac{\log \tilde{\lambda}_{N}^{(d+1)}(x,k)}{N} = 0.
\]
\end{theorem}

Thus by Theorem \ref{mainB}, if the orbit determination problem includes the determination of the parameter $k$, the confidence ellipsoid defined in \eqref{conf-ellip-A} has one of the axes which shrinks slower than any exponential.
Since the uncertainties are the projection of the confidence ellipsoid on the direction of the parameters to be determined, in general the slow decay of this axis affects all the uncertainties, giving a lower bound to their speed of decay.
In Section \ref{sec:example} we consider an example for which we can prove a more precise asymptotic behaviour for $\tilde{\lambda}_{N}^{(d+1)}(x,k)$.

\begin{remark}
By the proof of Theorem \ref{mainB} we cannot exclude that $\tilde{\lambda}_{N}^{(d+1)}$ converges to a positive constant. However this would imply the failure of the orbit determination process, since the confidence ellipsoid wouldn't shrink to a point. 
\end{remark}

\begin{remark}
The methods used in Theorems \ref{mainA} and \ref{mainB} can be applied also to non-hyperbolic diffeomorphisms, showing a less than exponential decay of the uncertainties also in case (A). This problem was studied in \cite{maro1} for nominal solutions living on invariant curves of exact symplectic twist maps of the cylinder, for which a sharp estimate for the rate of decay of the uncertainties was proved.
\end{remark}

\subsection{Comparison with the numerical results in \cite{ssm,smil} }
Our results in Theorems \ref{mainA} and \ref{mainB} are consistent with the numerical estimates in \cite{ssm,smil}. The authors considered a classical model in Celestial Mechanics: the well known Chirikov Standard Map defined as $f_k:\mathbb{T}^2\rightarrow\mathbb{T}^2$, $f_k(x,y)=(\bar{x},\bar{y})$
\begin{equation*}
  \left\{
  \begin{array}{l}
    \bar{x} = x+\bar{y} \\
    \bar{y} = y- k\sin x.
  \end{array}
  \right.
  \end{equation*}

The data of an orbit determination process were produced adding a random Gaussian noise to the orbit with initial condition $(x_0,y_0)=(3,0)$ and $k=0.5$. This initial condition is close to the hyperbolic fixed point and is likely giving rise to a chaotic orbit. The differential corrections algorithm is then performed both in case (A) and (B).

Working in quadruple precision, numerical instability of the differential corrections occurs for the number of observations $N\sim 300$. For the same number of iterations it is computed the largest eigenvalue of the state transition matrix. A linear fit gives a Lyapunov indicator of $+0.086$. It represents the largest Lyapunov exponent of the solution to which the differential corrections converge, i.e. the largest Lyapunov exponent of the nominal value.

To get a comparison with our results in case (A), we apply Theorem \ref{mainA} to the Standard Map with initial condition given by the nominal value obtaining
\[
-\gamma_1=\gamma_2=\gamma_*=\gamma^*=0.086\, ,
\]
hence assuming that the largest Lyapunov exponent coincides with the Lyapunov indicator. Hence we expect the eigenvalues of the covariance matrix to shrink as
\begin{equation}\label{stim_theo}
\lambda^{(i)}_{N,0.5} \sim e^{-2\times 0.086 N}\, , \quad \text{for }\, i=1,2\, .
\end{equation}
In the numerical simulations for case (A) performed in \cite{ssm,smil}, it was computed the standard deviation of the components $x$ and $y$ at every iteration, corresponding to the values $\sigma_x$ and $\sigma_y$ in Figure \ref{figure:ellipse}. By a linear fit in logarithmic scale, the authors got the slopes $-0.084$ for $x$ and $-0.083$ for $y$, deducing numerically the following decay of the uncertainties as function of $N$
\begin{equation}\label{stim_num}
\sigma_{x}(N)\sim e^{-0.084 N}, \qquad \sigma_{y}(N)\sim e^{-0.083 N}.
\end{equation}
Note that since the uncertainties $\sigma_x,\sigma_y$ are proportional to the square root of the eigenvalues $\lambda^{(i)}_{N,0.5}$ of the covariance matrix, the estimate \eqref{stim_theo} coming from Theorem \ref{mainA} is in perfect accordance with the numerical results \eqref{stim_num} obtained in \cite{ssm,smil}. Thus Theorem \ref{mainA} represents a proof of the following conjecture posed in \cite{smil} for case (A): ``{\em exponentially improving determination of the initial conditions only is possible, and the exponent appears to be very close to the opposite of the Lyapunov exponent.}''

Concerning case (B), the numerical simulations give a decrease of the uncertainties of the form $N^\alpha$, with different values of $\alpha<-0.5$ for the parameters $x,y,k$. No quantitative conjecture is posed on the rate of decrease apart from being strictly less that exponential. This is consistent with our results in Theorem \ref{mainB}, and with the lower bound that we obtain for the decay of the uncertainties for the systems studied in Section \ref{sec:example}.

\section{Proof of Theorem \ref{mainA}}\label{sec:proofA}

Let us fix $k\in K$ and let $x\in X$ be a point for which Oseledets Theorem holds. Consider the normal matrix $C_{N,k}(x)$ as in \eqref{normmat} and, recalling that it is positive definite, denote by $0<\delta_{N,k}^{(1)}(x)\leq\dots\leq\delta_{N,k}^{(d)}(x)$ its eigenvalues. Using Oseledets Theorem, for every $i=1,\dots,r$ let $E_i$ be the vector space of the decomposition of $T_xX$, and write for a unit vector $v\in E_i$ 
\begin{equation}\label{boh1}
v^TC_{N,k}(x)v = \sum_{|n|\leq N}v^TF_k^n(x)^TF_k^n(x)v = \sum_{|n|\leq N}|F_k^n(x)v|^2.
\end{equation}
By Oseledets Theorem and \eqref{limiti-lyap}, for every $\varepsilon>0$ there exist $n_+=n_+(x)>0$ and $n_-=n_-(x)<0$ such that
  \begin{align*}
    &e^{(\gamma_i-\varepsilon)|n|}<|F_k^n(x)v|<e^{(\gamma_i+\varepsilon)|n|} \quad&  \mbox{for all }n>n_+, \\
    &e^{(-\gamma_i-\varepsilon)|n|}<|F_k^n(x)v|<e^{(-\gamma_i+\varepsilon)|n|} \quad&  \mbox{for all }n<n_-,
  \end{align*}
hence from \eqref{boh1} we can choose $\varepsilon\in (0,\gamma_*)$ and find $\bar{N}(x):=\max\{n_+(x),-n_-(x)\}$ such that for all $N>\bar{N}(x)$
\[
c_-(x) + \sum_{n=0}^Ne^{2(|\gamma_i|-\varepsilon)n}  < v^TC_{N,k}(x)v < c_+(x) + \sum_{n=0}^Ne^{2(|\gamma_i|+\varepsilon)n} + \sum_{n=0}^Ne^{2(-|\gamma_i|+\varepsilon)n}
\]
where
  \[
c_\pm(x) = \sum_{n=1}^{n_+(x)}\left(|F_k^n(x)v|^2 - e^{2(|\gamma_i|\pm \varepsilon)n} \right)  + \sum_{n=n_-(x)}^{-1}\left(|F_k^{n}(x)v|^2 - e^{2(-|\gamma_i|\pm \varepsilon)|n|} \right)
  \]
and in the left hand side we have neglected the terms with $e^{2(-|\gamma_i|-\varepsilon)|n|}$ for which the series converge. Hence, recalling the variational characterisation of the eigenvalues of a simmetric matrix, given $\varepsilon\in (0,\gamma_*)$ there exists $\bar{N}(x)$ such that for $N>\bar{N}(x)$
  \begin{align*}
    \delta_{N,k}^{(1)}(x) &= \min_{v\in T_x X, |v|=1}v^TC_{N,k}(x)v > c_-(x) + \sum_{n=0}^Ne^{2(\min_i |\gamma_i|-\varepsilon)n} \\
    & = c_-(x) + \sum_{n=0}^Ne^{2(\gamma_*-\varepsilon)n} = c_-(x) +\frac{e^{2(\gamma_*-\varepsilon)(N+1)}-1}{e^{2(\gamma_*-\varepsilon)}-1} \, .
  \end{align*}
Analogously for $N>\bar{N}(x)$, using $\tilde c_+(x):= c_+(x) + \max_i \sum_{n=0}^{+\infty}e^{2(-|\gamma_i|+\varepsilon)n}$,
 \begin{align*}
   \delta_{N,k}^{(d)}(x) &= \max_{v\in T_x X, |v|=1}v^TC_{N,k}(x)v < \tilde{c}_+(x) + \sum_{n=0}^Ne^{2(\max_i|\gamma_i|+\varepsilon)n} \\
   & = \tilde{c}_+(x) + \sum_{n=0}^N e^{2(\gamma^*+\varepsilon)n} = \tilde{c}_+(x) +\frac{e^{2(\gamma^*+\varepsilon)(N+1)}-1}{e^{2(\gamma^*+\varepsilon)}-1}
 \end{align*}
Finally, using that by definition  
\[
\lambda_{N,k}^{(1)}(x) = [\delta_{N,k}^{(d)}(x)]^{-1}, \quad \lambda_{N,k}^{(d)}(x) = [\delta_{N,k}^{(1)}(x)]^{-1} 
\] 
we have that for every $\varepsilon\in (0,\gamma_*)$ 
\[
-2(\gamma^*+\varepsilon) \le \liminf_{N\to +\infty}\frac{\log \lambda_{N,k}^{(i)}(x)}{N} \le \limsup_{N\to +\infty}\frac{\log \lambda_{N,k}^{(i)}(x)}{N} \le -2 (\gamma_*-\varepsilon)\, .
\]
Since $\varepsilon$ is arbitrary, the theorem follows. \qed

\section{Proof of Theorem \ref{mainB}} \label{sec:proofB}

Let us introduce the auxiliary diffeomorphism
\[
g:X\times K \rightarrow X\times K, \quad g(x,k)=(f_k(x),k).
\]
Recalling \eqref{Fktilde}, we can write the $(d+1)\times(d+1)$ Jacobian matrix $G(x,k)$ of $g$ with respect to $(x,k)$ as follows
\begin{equation}
G(x,k) = \left( \begin{array}{ccc|c} \frac{\partial (f_k)_1}{\partial x_1}(x) & \dots & \frac{\partial (f_k)_1}{\partial x_d}(x) & \frac{\partial (f_k)_1}{\partial k}(x) \\
    \vdots &  & \vdots & \vdots \\
    \frac{\partial (f_k)_d}{\partial x_1}(x) & \dots & \frac{\partial (f_k)_d}{\partial x_d}(x) & \frac{\partial (f_k)_d}{\partial k}(x)\\
\hline
0 &\dots  & 0 & 1
\end{array}
\right)
=
\left(
  \begin{array}{ccc|c}
     \multicolumn{4}{c}{\tilde{F}_k(x)} \\
    \hline
    0 &\dots  & 0 & 1 
   \end{array}
  \right)
  \end{equation}
and for $n\in\ZZ$ we denote by $G^n(x,k)$ the Jacobian matrix of $g^n$ with respect to $(x,k)$. As in \eqref{Fnk}, by the chain rule it can be written as
\begin{equation}
  G^n(x,k) =\left\{
  \begin{split}
   &  G(g^{n-1}(x,k))G(g^{n-2}(x,k))\dots G(x,k) \quad &\mbox{for }n\geq 1, \\
    & \mathbbm{1} \quad &\mbox{for }n = 0, \\ 
    & G^{-1}(g^{n}(x,k))G^{-1}(g^{n+1}(x,k))\dots G^{-1}(g^{-1}(x,k)) \quad &\mbox{for }n<0. 
    \end{split}
  \right. .
  \end{equation}
Finally we consider the auxiliary normal matrix
\[
C_N^g(x,k) = \sum_{|n|\leq N}G^n(x,k)^TG^n(x,k),
\]
which is related to the normal matrix $\tilde{C}_N(x,k)$ as shown in the following lemma.
\begin{lemma}\label{lem1}
  For all $n\in \ZZ$ it holds 
\begin{equation}\label{thlem1}
G^n(x,k) =\left(
  \begin{array}{ccc|c}
     \multicolumn{4}{c}{\tilde{F}^n_k(x)} \\
    \hline
    0 &\dots  & 0 & 1    
  \end{array}
  \right)
  =
  \left(
 \begin{array}{ccc|c}
     \multicolumn{3}{c|}{F^n_k(x)} & \frac{\partial f_k^n}{\partial k}(x) \\
    \hline
    0 &\dots  & 0 & 1    
  \end{array}
  \right)  
\end{equation}
so that for all $N\ge 1$
\begin{equation}\label{thlem2}
  C^g_N(x,k)=\tilde{C}_N(x,k) +
  \left(
  \begin{array}{ccc|c}
    & & & 0\\
     \multicolumn{3}{c|}{0} & \vdots \\
& & & 0\\
     \hline
    0 &\dots  & 0 & 2N+1   
  \end{array}
  \right)
  \end{equation}
\end{lemma}
\begin{proof}
  Formula \eqref{thlem1} comes from the definitions noting that for $n>0$
  \[
\frac{\partial f_k^n}{\partial k}(x) = \frac{\partial}{\partial k}(f_k(f_k^{n-1}(x))=  \frac{\partial f_k}{\partial k}(f_k^{n-1}(x))+F_k(f_k^{n-1}(x))\frac{\partial f_k^{n-1}}{\partial k}(x)
\]
and a similar formula holds for $n<0$. Formula \eqref{thlem2} is a straightforward consequence of \eqref{thlem1}.
\end{proof}

We now study the Lyapunov exponents of $g$. First we consider the measure $\mu \times \delta_k$ on $X\times K$ which is clearly $g$-invariant. Moreover, if $f_k$ is ergodic the same holds for $g$. Thus, under assumption \eqref{L1} for $f_k$ we can apply Oseledets Theorem to $g$, and obtain that for $\mu$-almost every $x\in X$ the map $g$ admits Lyapunov exponents $\tilde{\gamma}_1,\dots,\tilde{\gamma}_{\tilde{r}}$ and an associated decomposition
\[
T_x(X\times K) = \tilde{E}_1\oplus \tilde{E}_2\oplus \dots\oplus \tilde{E}_{\tilde{r}}\, .
\]
In the following lemma we describe the relation between the Lyapunov exponents of $g$ and those of $f_k$.

\begin{lemma}\label{lem2}
Given $g$ as above, we have $\tilde{r}=r+1$, and
\begin{equation*}
  \left\{
  \begin{aligned}
    \tilde{\gamma}_i = \gamma_i, \quad &\dim \tilde{E}_i =\dim E_i  \quad \mbox{for } i=1,\dots, r \\
    \tilde{\gamma}_{\tilde{r}}= 0,  \quad &\dim \tilde{E}_{\tilde{r}} =1. 
    \end{aligned}
    \right.
  \end{equation*}

\end{lemma}
\begin{proof}
First of all, if $\gamma_i$ is a Lyapunov exponent of $f_k$ then it is also a Lyapunov exponent of $g$ with the same multiplicity, in the sense that $\tilde{\gamma}_i= \gamma_i$ and $\dim E_i =\dim \tilde{E}_i$ for all $i=1,\dots,r$. Actually, for every $v\in\ E_i$ and for $\mu$-almost every $x\in X$, using \eqref{thlem1} we have
  \begin{equation}
    \gamma_i=\lim_{n\to\pm\infty}\frac{1}{n}\log \left|F^n_k(x)v\right| = \lim_{n\to\pm\infty}\frac{1}{n}\log \left|G^n(x,k)(v,0)^T\right| = \tilde{\gamma}_i. 
    \end{equation}
Then we prove that the last exponent of $g$ is zero. To this end, we recall that by Oseledets Theorem the eigenvalues of the matrix
  \[
\Lambda_k(x)=\lim_{n\to\infty}(F^n_k(x)^TF^n(x))^{1/2n}
  \]
are $e^{\gamma_1}, e^{\gamma_2},\dots, e^{\gamma_{r}}$ with multiplicity $\dim E_i$ and the eigenvalues of the matrix
  \[
\tilde{\Lambda}(x,k)=\lim_{n\to\infty}(G^n(x,k)^TG^n(x,k))^{1/2n}
  \]
are $e^{\tilde{\gamma}_1}, e^{\tilde{\gamma}_2},\dots, e^{\tilde{\gamma}_{\tilde{r}}}$ with multiplicity $\dim \tilde{E}_i$. Now, by \eqref{thlem1}, for every $n\in\ZZ$
  \[
\det(G^n(x,k)^TG^n(x,k)) = (\det F_k^n(x) )^2 = \det( F_k^n(x)^T F_k^n(x))
\]
so that $\det\tilde{\Lambda}(x,k) = \det\Lambda_k(x)$ that is
\[
\sum_{i=1}^{\tilde{r}}\tilde{\gamma}_i \, \dim \tilde{E}_i = \sum_{i=1}^{r}\gamma_i\, \dim E_i. 
\]
But since the Lyapunov exponents of $f_k$ are also Lyapunov exponents of $g$ with the same multiplicity, we must have $\tilde{r}=r+1$, $\tilde{\gamma}_{\tilde{r}}=0$ and $\dim \tilde{E}_i=1$.
\end{proof}
The conclusion of the proof of Theorem \ref{mainB} is a consequence of the following lemma. To state it, let us consider the normal matrix
\[
\tilde{C}_{N}(x,k) = \sum_{|n|\leq N}\tilde{F}_{k}^n(x)^T\tilde{F}_{k}^n(x)
\]
and denote its eigenvalues by
\[
0<\tilde{\delta}_{N}^{(1)}(x,k)\leq\dots\leq\tilde{\delta}_{N}^{(d+1)}(x,k).
\]
\begin{lemma}\label{lem3}
For $\mu$-almost every $x\in X$ the smallest eigenvalue $\tilde{\delta}_{N}^{(1)}(x,k)$ of $\tilde{C}_{N}(x,k)$ is a positive number which increases with $N$ and satisfies
\[
\lim_{N\to+\infty}\frac{\log \tilde{\delta}_{N}^{(1)}(x,k)}{N} =0.
\]
\end{lemma}

\begin{proof}
First $\{\tilde{\delta}_{N}^{(1)}(x,k)\}_N$ is an increasing sequence of positive terms since 
\[
\tilde{\delta}_{N}^{(1)}(x,k) = \min_{v\in T_x (X\times K), |v|=1}v^T\tilde{C}_{N}(x,k)v 
\]
and $v^T\tilde{C}_{N}(x,k)v$ is the sum of $2N+1$ positive terms. 

Let us now denote by $v_0\in \RR^{d+1}$ the unit vector corresponding to the vanishing Lyapunov exponent of $g$, so that
\begin{equation} \label{conv-zero}
\lim_{n\to\pm \infty}\frac{1}{|n|}\log\left|G^n(x,k)v_0 \right| =0.
\end{equation}
Hence, for every $\varepsilon>0$ there exists $\bar{n}$ such that  $\left|G^n(x,k)v_0 \right|^2<e^{2\varepsilon |n|}$ for $|n|>\bar{n}$,  so that, for $N>\bar{n}$
\[
\begin{aligned}
  \min_{v\in T_x(X\times K), |v|=1}v^TC^g_{N}(x,k)v  & \leq v_0^TC^g_{N}(x)v_0 =\sum_{|n|\leq N}\left|G^n(x,k)v_0 \right|^2 \\
  & = \sum_{|n|\leq \bar{n}}\left|G^n(x,k)v_0 \right|^2 + \sum_{\bar{n}<|n|\le N}\left|G^n(x,k)v_0 \right|^2 \\
  & \leq c(x,k) +2\sum_{n=0}^N e^{2\varepsilon n} = c(x,k) + 2\frac{e^{2\varepsilon (N+1)}-1}{e^{2\varepsilon } -1},  
\end{aligned}
\]
for some constant $c(x,k)$ independent on $N$. 

We now use Lemma \ref{lem1} to find an estimate for the eigenvalues of $\tilde{C}_{N}(x,k)$. From the variational characterisation of the eigenvalues and \eqref{thlem2},
\[
\begin{aligned}
\tilde{\delta}_{N}^{(1)}(x,k) & = \min_{v\in T_x(X\times K), |v|=1}v^T\tilde{C}_{N}(x,k)v  \leq 
\left(\min_{v\in\RR^{d+1}, |v|=1}v^TC^g_{N}(x,k)v\right) +2N+1 \\
& \leq c(x,k) + 2\frac{e^{2\varepsilon (N+1)}-1}{e^{2\varepsilon } -1} +2N+1.
\end{aligned}
\]
Hence, for every $\varepsilon>0$ there exists a constant $c_1(x,k)$ not depending on $N$ such that
\[
0\le \lim_{N\to +\infty} \frac{\log \tilde{\delta}_{N}^{(1)}(x,k)}{N}\le \lim_{N\to +\infty} \frac{2\varepsilon(N+1) + \log(c_1(x,k)+\log N)}{N} = 2\varepsilon\, .
\]
Since $\varepsilon$ is arbitrary the result is proved.
\end{proof}

To finish the proof of Theorem \ref{mainB} is now enough to recall that $\tilde{\lambda}_{N}^{(d+1)}(x,k) = [\tilde{\delta}_{N}^{(1)}(x,k)]^{-1}$ and apply Lemma \ref{lem3}.
\qed

\section{An example}\label{sec:example}

In this section we present a class of maps for which the estimates on the eigenvalues of $\Gamma_{N,k}$ and $\tilde{\Gamma}_N$ in Theorems \ref{mainA} and \ref{mainB} can be made explicit. Inspired by some computations presented in \cite{bm} we consider the case of an affine hyperbolic diffeomorphism $\mathcal{C}_k: \TT^d\rightarrow \TT^d$ of the $d$-dimensional torus $\TT^d = \RR^d/\ZZ^d$. One example of this class for $d=2$ is the well-known Arnold's Cat Map.

Fixing a matrix $A\in SL(d,\ZZ)$ and a vector $b\in\RR^d$ we define 
\begin{equation}
  \mathcal{C}_k(x) = Ax +kb\, .
\end{equation}
Since $\det A =1$ the Lebesgue measure $m$ is $\mathcal{C}_k$-invariant. Finally we assume that $A$ has no eigenvalues of modulus 1, since as shown below this implies that $\mathcal{C}_k$ is hyperbolic. We denote by $\delta_1,\dots,\delta_d$ the eigenvalues  of $A$.

The orbits of the map $\mathcal{C}_k$ can be computed explicitly, more precisely, we have
\begin{lemma}\label{n_cat}
For every $n\in\mathbb{Z}$, setting $w=(\mathbbm{1}-A)^{-1}b$,
  \[
  \mathcal{C}^n_k(x) = A^n x
   + k(\mathbbm{1}-A^n)w, \quad \frac{\partial \mathcal{C}^n_k}{\partial k}(x) = (\mathbbm{1}-A^n)w 
  \]
  for all $x\in \TT^d$.
  \end{lemma}
\begin{proof}
  The case $n=0$ is trivial. For $n\neq 0$ it follows directly noting that for $n>0$
\begin{align*}
  \mathcal{C}^n_k(x) &=    A^n x  
 +k\sum_{i=0}^{n-1} 
 A^ib =
 A^n x
   + k(\mathbbm{1}-A^n)(\mathbbm{1}-A)^{-1}b, 
\\
\mathcal{C}^{-n}_k(x,y) &=    A^{-n} x
 -k\sum_{i=1}^{n} 
 A^{-i}b =
 A^{-n} x  + k(\mathbbm{1}-A^{-n})(\mathbbm{1}-A)^{-1}b.
\end{align*}
\end{proof}

It is an easy consequence of this lemma that the matrices $F^n_k(x)$ and $\tilde{F}^n_k(x)$ introduced in \eqref{Fktilde} and \eqref{Fnktilde} are constant and independent on $x$ and $k$. For the same reason, the Lyapunov exponents of $\mathcal{C}_k$ are constant everywhere. 

We now give Theorem \ref{mainA} for maps $\mathcal{C}_k$ under the assumption that $A$ is symmetric. In this case the result is much sharper, giving the exact exponential rate of decrease for all the eigenvalues of the covariance matrix $\Gamma_{N,k}(x)$.

\begin{proposition}\label{caseAcat}
Let $\mathcal{C}_k: \TT^d\rightarrow \TT^d$ be defined as above with $A$ symmetric, and let $\gamma_1,\dots,\gamma_d$ its Lyapunov exponents counted with multiplicity, that is the exponents are not necessarily different. Then the eigenvalues $\lambda_{N,k}^{(i)}$ of the covariance matrix $\Gamma_{N,k}$ satisfy
\[
\lim_{N\to +\infty} \frac{\lambda_{N,k}^{(i)}}{e^{-2|\gamma_i|N}} = \left\{ \begin{array}{ll}
       1-|\delta_i|^{-2}\, , &\mbox{if }|\delta_i|>1, \\[0.3cm]
       1-|\delta_i|^{2}\, ,   &\mbox{if }|\delta_i|<1,
\end{array} \right. \qquad \mbox{for }i=1,\dots,d.
\]
\end{proposition}
\begin{proof}
Since $A$ is symmetric, there exists an orthonormal matrix $P$ such that
\begin{equation}\label{A_diag}
A= P^T\Lambda P, \qquad \Lambda =\diag (\delta_1,\dots,\delta_d), \qquad P^TP=\mathbbm{1},
\end{equation}
with $\delta_i \in \RR$ for all $i=1,\dots,d$, and in particular the Lyapunov exponents of $\mathcal{C}_k$ are given by $\gamma_i = \log |\delta_i|$. Since $A$ has no eigenvalues of modulus 1, the map $\mathcal{C}_k$ is hyperbolic (see Definition \ref{def:hyp-map}).

Now, from \eqref{A_diag}, the normal matrix satisfies
\[
C_{N,k} = \sum_{|n|\leq N}(A^n)^T A^n =P^T \sum_{|n|\leq N} \diag (\delta_1^{2n},\dots,\delta_d^{2n})P
\]
and its eigenvalues are
\[
  \delta_{N,k}^{(i)} =
    \left\{ \begin{array}{ll}
       \frac{\delta_i^{2N}}{1-\delta_i^{-2}} (1 +O(\delta_i^{-2N}))\, , &\mbox{if }|\delta_i|>1, \\[0.3cm]
      \frac{\delta_i^{-2N}}{1-\delta_i^{2}} (1 +O(\delta_i^{2N}))\, ,   &\mbox{if }|\delta_i|<1,
\end{array}\right.
\]
where we recall that the eigenvalues $\delta_i$ are real. We conclude using that  
\[
\Gamma_{N,k} = C^{-1}_{N,k} 
\]
and $\gamma_i = \log|\delta_i|$.
\end{proof}

We can say more also for the asymptotic behaviour of the largest eigenvalue $\tilde{\lambda}_N^{(d+1)}$ of the covariance matrix $\tilde{\Gamma}_N$ of the orbit determination problem in case (B). In Theorem \ref{mainB} we proved that $\tilde{\lambda}_N^{(d+1)}$ decreases slower than exponentially, the lack of a precise estimate being due to the uncertainty on the speed of convergence to zero in \eqref{conv-zero}. In the class of maps we are studying in this section, we can be much more precise on the asymptotic behaviour of $\tilde{\lambda}_N^{(d+1)}$.

\begin{proposition}\label{caseBcat}
Let $\mathcal{C}_k: \TT^d\rightarrow \TT^d$ be defined as above. The largest eigenvalue $\tilde{\lambda}_{N}^{(d+1)}$ of the covariance matrix $\tilde{\Gamma}_{N}$ for case (B) of the orbit determination problem satisfies
\[
\tilde{\lambda}_{N}^{(d+1)} \ge \frac{|w|^2+1}{|w|^2}\, \frac{1}{2N+1}
\]
for all $N\ge 1$.
\end{proposition}
\begin{proof}
Using the same notation of Section \ref{sec:proofB}, we consider the auxiliary map $g:\TT^d\times K \to \TT^d\times K$, $g(x,k)=(\mathcal{C}_k(x),k)$, and recalling \eqref{thlem1} and Lemma \ref{n_cat}, its Jacobian matrix $G^n$ takes for all $n\in \ZZ$ the form
\[
G^n(x,k) =   \left( \begin{array}{ccc|c}
    \multicolumn{3}{c|}{A^n} & (\mathbbm{1}-A^n)w \\
    \hline
    0 &\dots  & 0 & 1    
  \end{array}
  \right).
\]
for all $(x,k)$, where we recall that $w =(\mathbbm{1}-A)^{-1}b$ . We note that the eigenvalues of $G^n$ are equal to $\delta_1^{n},\dots,\delta_d^n,1$, where $\delta_1,\dots,\delta_d$ are the eigenvalues of $A$. Actually, choosing $v_i$ such that $Av_i = \delta_iv_i$, then $G^n(v_i,0)^T = (A^nv_i,0)^T=\delta_i^n(v_i,0)^T$, and, choosing the vector $v_0=(w,1)\in\RR^{d+1}$ we have
  \begin{align*}
G^nv_0 &= \left(
 \begin{array}{ccc|c}
     \multicolumn{3}{c|}{A^n} & (\mathbbm{1}-A^n)w \\
    \hline
    0 &\dots  & 0 & 1    
  \end{array}
 \right)
\left(
 \begin{array}{c}
    w \\
    1    
  \end{array}
  \right) \\
  &=
  \left(
 \begin{array}{c}
    A^nw+ (\mathbbm{1}-A^n)w \\
    1    
  \end{array}
  \right)
  =
  \left(
 \begin{array}{c}
    w \\
    1    
  \end{array}
  \right) =v_0.
  \end{align*}
We thus get that for the normal matrix $C^g_{N}(x,k)$ it holds
\[
v_0^T C^g_{N}(x,k) v_0 = \sum_{|n|\le N} |G^n v_0|^2 = (2N+1) |v_0|^2
\]
and for the normal matrix $\tilde{C}_{N}(x,k)$ of $\mathcal{C}_k$ relative to case (B) of the orbit determination problem, we use \eqref{thlem2} to write
\[
v_0^T \tilde{C}_{N}(x,k) v_0 = v_0^T C^g_{N}(x,k) v_0 - (2N+1) = (2N+1) (|v_0|^2-1) = (2N+1) |w|^2.
\]
Hence, from the variational characterisation of the eigenvalues of $\tilde{C}_{N}(x,k)$, we obtain that its smallest eigenvalue $\tilde{\delta}_N^{(1)}$ satisfies
\[
\tilde{\delta}_N^{(1)} = \min_{v\in\RR^{d+1}, |v|=1}v^T\tilde{C}_{N}(x,k)v \le \frac{1}{|v_0|^2}\, v_0^T \tilde{C}_{N}(x,k) v_0 = (2N+1) \frac{|w|^2}{|w|^2+1} \\
\]
and the result follows recalling that   
\[
\tilde{\lambda}_{N}^{(d+1)} = [\tilde{\delta}_{N}^{(1)}]^{-1}.
\]
\end{proof}

\section{Conclusions and future work}\label{sec:conclusions}
We have considered the problem of orbit determination under the assumption that the
number of observations grows simultaneously with the time span over
which they are performed. Following the numerical results in \cite{ssm,smil} we have studied the asymptotic rate of decay of the uncertainties as the number of observations grows.

We have considered the problem for hyperbolic maps, for which all the Lyapunov exponents are
not zero, depending on a parameter and we have treated separately the cases in which the
parameter is included or not in the orbit determination procedure. We have
analitically proved that if the parameter is not included then the
uncertainties decrease exponentially, while if the parameter is
included then the uncertainties decrease strictly slower than
exponentially.  This is consistent with the numerical results and
gives a proof of one of the main questions posed in \cite{ssm,smil}.

Together with the results in \cite{maro1}, which considered the ordered
case (KAM scenario), this paper is a step forward the complete
understanding of the numerical results.

\section*{Acknowledgements}
  This paper is dedicated to the memory of Andrea Milani who suggested the problem. We thank Gianluigi Del Magno for discussions on the proof of Lemma \ref{lem2}. We would like to thank also the unknown referees for several valuable advice that improved the final version of the paper.

\end{document}